\documentclass[10pt]{article}

\usepackage{amsfonts}
\usepackage{amsmath}
\usepackage{amssymb}
\usepackage{amsthm}
\usepackage{bm}
\usepackage{enumerate}
\usepackage{color}

\newcommand{\Z}{\mathbb{Z}}

\newcommand{\C}{\mathbb{C}}

\newcommand{\bu}{\bm{u}}
\newcommand{\bv}{\bm{v}}

\newtheorem{theorem}{Theorem}[section] % 1st argument is your name for it
\newtheorem{lemma}[theorem]{Lemma}     % 2nd argument is what is printed

\newtheorem{definition}[theorem]{Definition}

\title{Generalised Elliptic Functions}

\author{M.England and C.Athorne}

\date{October 2011}

\begin{document}

\maketitle

\begin{abstract}
We consider multiply periodic functions, sometimes called
Abelian functions, defined with respect to the period matrices
associated with classes of algebraic curves.  We realise them as
generalisations of the Weierstra{\ss} $\wp$-function using two
different approaches.  These functions arise naturally as solutions
to some of the important equations of mathematical physics and their
differential equations, addition formulae, and applications have all
been recent topics of study.

The first approach discussed sees the functions defined as
logarithmic derivatives of the $\sigma$-function, a modified Riemann
$\theta$-function.  We can make use of known properties of the sigma
function to derive power series expansions and in turn the
properties mentioned above.  This approach has been extended to a
wide range of non hyperelliptic and higher genus curves and an
overview of recent results is given.

The second approach defines the functions algebraically, after first
modifying the curve into its equivariant form. This approach allows
the use of representation theory to derive a range of results at
lower computational cost.  We discuss the development of this theory
for hyperelliptic curves and how it may be extended in the future.
\end{abstract}

\section{Introduction} \label{SEC_Intro}
In this paper we work with generalisations of the Weierstra{\ss} $\wp$-function.
This is an elliptic function, and so periodic with respect to two independent complex periods $\omega_1,\omega_2$:
\begin{equation} \label{eq:Intro_period}
\wp(u + \omega_1) = \wp(u + \omega_2) = \wp(u), \qquad \mbox{for all } \, u \in \mathbb{C}.
\end{equation}
The $\wp$-function has the simplest possible pole structure for an elliptic function and satisfies many interesting properties.
For example, the $\wp$-function can be used to parametrise an elliptic curve,
\begin{equation} \label{eq:Intro_ec}
y^2 = 4x^3 - g_2x - g_3,
\end{equation}
where $g_2$ and $g_3$ are constants.  It also satisfies the following well-known differential equations,
\begin{eqnarray}
\big(\wp'(u)\big)^2 &=& 4\wp(u)^3 - g_2\wp(u) - g_3,            \label{eq:Intro_elliptic_diff1}     \\
           \wp''(u) &=& 6\wp(u)^2 - \textstyle \frac{1}{2}g_2.  \label{eq:Intro_elliptic_diff2}
\end{eqnarray}
Elliptic functions have been the subject of much study since their discovery and have been extensively used to enumerate solutions of non-linear wave equations.  They occur in many physical applications; traditionally the arc-length of the lemniscate and the dynamics of spherical pendulums, \cite{MKandM99}, but also in cryptography, \cite{Washington08}, and soliton solutions to the KdV equation, \cite{Dodd82}.

Recent times have seen a revival of interest in the theory of their generalisations,
the Abelian functions, which have multiple independent periods, or more accurately,
are periodic with respect to a multi-dimensional period lattice.  The lattice is usually defined
in association with an underlying algebraic curve.  These functions are also
beginning to find a wide range of applications. For example, they give further
solutions to the KdV equation along with solutions to other integrable equations
from the KP-hierarchy, (see for example \cite{bel97}, \cite{bego08} \cite{MEe09}). They have also been used to describe geodesic motions in
certain space-time metrics, \cite{EHKKL}.

This introduction will continue with a discussion of the motivation
to study such functions followed by a reminder of the Riemann-Roch
theorem and its implications in this context.  Then in
Section \ref{SEC_Sigma} we discuss how the theory of higher genus
functions can be derived by defining them via an auxiliary function,
the $\sigma$-function.  We describe how a series expansion
of this function can be derived and the results which follow.  In
Section \ref{SEC_Equivar} we proceed to give an alternative
definition of the functions, developed in recent years. This
algebraic definition and the following approach is termed
equivariant and brings out an ${\mathfrak sl}_2(\mathbb C)$
structure in the equations which can be exploited to reduce the
computation involved in working with the functions. Finally in
Section \ref{SEC_Future} we discuss the possible future directions
of this research.

\subsection{Motivation}

As discussed above, the Weierstra{\ss} $\wp$-function and its generalisations have strong connections with integrable PDEs.
The KdV equation has a translational symmetry reduction described by
the Weierstra{\ss} $\wp$-function and hence is associated with a genus
one curve. More generally the KdV with other specific, differential
constraints is associated with generalised $\wp$-functions and genus
$g$ hyperelliptic curves. For details see the example below. Of
course, such solutions also arise via the Integrable Scattering
Transform and are expressed in terms of $\theta$-functions. So the
interest for Integrable Systems is, in part, to obtain equations
satisfied by $\wp$-functions for general plane algebraic curves. On
a wider plain there is interest in application to problems in GR,
Billiards, Statistical mechanics and elsewhere. Most fundamentally
we are concerned here with the interplay between function theory and
the geometry of curves.

\subsubsection*{Example: A genus two curve}
Consider for $x,y\in \C$, the genus two curve
\[
y^2 = \lambda_6x^6 + 6\lambda_5x^5 + 15\lambda_4x^4 + 20\lambda_3x^3 + 15\lambda_2x^2 + 6\lambda_1x + \lambda_0.
\]
In this case there are generalised $\wp$-functions, denoted $\wp_{ij}(u_1,u_2)$ for $i,j=1,2$.
They are functions of $g=2$ variables and are multiply periodic.
We can denote the derivatives of the functions using additional indices, so for example
\[
\frac{\partial}{\partial u_i} \wp_{jk} = \wp_{ijk}.
\]
For simplicity we shorten the differentiation notation so that
\[
\partial_i = \frac{\partial}{\partial u_i}
\]
for example.  It is the case that $\wp_{ij}=\wp_{ji}$ and so we have three generalised $\wp$-functions.  We also have
\[
\partial_i \wp_{jk} = \partial_j\wp_{ik}.
\]
and so generally the order of the indices is not relevant.  We adopt the notation of writing the indices in ascending numerical order.

These functions can be shown to satisfy the following set of partial differential equations.
\begin{align}
\textstyle -\frac13\wp_{2222}+2\wp_{22}^2
&= \textstyle \lambda_4\wp_{22}-2\lambda_5\wp_{12}+\lambda_6\wp_{11}+\lambda_2\lambda_6-4\lambda_3\lambda_5+3\lambda_4^2 \nonumber \\
\textstyle -\frac13\wp_{1222}+2\wp_{12}\wp_{22}
&= \textstyle \lambda_3\wp_{22}-2\lambda_4\wp_{12}+\lambda_5\wp_{11}+\frac12(\lambda_1\lambda_6-3\lambda_2\lambda_5+2\lambda_3\lambda_4) \nonumber \\
\textstyle -\frac13\wp_{1122}+\frac23\wp_{11}\wp_{22}+\frac43\wp_{12}^2
&= \textstyle \lambda_2\wp_{22}-2\lambda_3\wp_{12}+\lambda_4\wp_{11}+\frac16(\lambda_0\lambda_6-9\lambda_2\lambda_4+8\lambda_3^2) \nonumber \\
\textstyle -\frac13\wp_{1112}+2\wp_{11}\wp_{12}
&= \textstyle \lambda_1\wp_{22}-2\lambda_2\wp_{12}+\lambda_3\wp_{11}+\frac12(\lambda_0\lambda_5-3\lambda_1\lambda_4+2\lambda_2\lambda_3) \nonumber \\
\textstyle -\frac13\wp_{1111}+2\wp_{11}^2
&= \textstyle \lambda_0\wp_{22}-2\lambda_1\wp_{12}+\lambda_2\wp_{11}+\lambda_0\lambda_4-4\lambda_1\lambda_3+3\lambda_2^2 \nonumber
\end{align}
Consider the first equation.  By a rational transformation we may, without loss of generality, set $\lambda_6=0,\, \break \lambda_4=0,\,\lambda_5=\frac{1}{6}$ to give
\[
-\frac{1}{3}\wp_{2222} + 2\wp_{22}^2 = - \frac{1}{3}\wp_{12} - \frac{2}{3}\lambda_3.
\]
We differentiate the identity with respect to $u_2$ to find
\[
-\frac{1}{3}\wp_{22222} + 4\wp_{22}\wp_{222} = -
\frac{1}{3}\wp_{122},
\]
and making the change of variables
$\partial_1=\partial_t,\,\partial_2=\partial_x\,\wp_{22}=U(x,t)$
this is just
\[
U_t - U_{xxx}+12UU_x=0.
\]
Thus we have the KdV-equation, along with four additional differential identities.

We can write these five partial differential equations in a parametrised form.  Thus the relations on 4-index objects are
\begin{equation} \label{eq:g2_4indexp_Para}
\textstyle -\frac13\partial^4\wp+2(\partial^2\wp)^2
= \textstyle H_0\bar\partial^2\wp-2H_1\bar\partial\partial\wp+H_2\partial^2\wp+G
\end{equation}
where $\partial=\partial_1+\nu\partial_2$, $\bar\partial=\partial_2,$ the $H_i$ are functions of $\lambda_i$ and $G$ is a function of $\lambda_i$ and $\nu$. Another set of identities in the function are those quadratic in the $3$-index objects.  These may be written in the parametrised form
\[
\textstyle
\partial^3\wp\bar\partial^3\wp-6(\partial^2\wp)(\partial\bar\partial\wp)(\bar\partial^2\wp)
=
H'(\bar\partial^2\wp,\bar\partial\partial\wp,\partial^2\wp)
\textstyle
\]
where $H'$ is a polynomial of degree two in the second derivatives of $\wp$ with coefficients in $\lambda_i$.

Equation (\ref{eq:g2_4indexp_Para}) is a reduction of the Boussinesq equation and so we can write down a Lax pair for the genus two $\wp$-function equations:
\begin{align*}
\textstyle L &= \zeta\bar\partial+\partial^2-2(\partial^2\wp) \textstyle \\
\textstyle M &= \textstyle \partial^3+\frac12\zeta'\partial^2+\frac{1}{10}(\zeta^2)'\partial-3(\partial^2\wp)\partial
\textstyle - \frac{3}{2}(\partial^3\wp)-\zeta'(\partial^2\wp)+\frac{3}{2}\zeta(\partial\bar\partial\wp) \textstyle
\end{align*}
where $\zeta^2=g(\nu)$ and $\zeta'=\frac{d\zeta}{d\nu}$, \cite{AEE03}.

\subsection{Riemann-Roch Theorem} \label{SubSec_RR}
Let $X$ be a compact, non-singular, complex algebraic curve (Riemann Surface) and $D$ a divisor on $X$:
\[
D = n_1P_1+n_2P_2+\ldots +n_rP_r, \qquad P_i\in X,\,n_i\in \Z.
\]
The Riemann-Roch theorem relates dimensions of spaces of functions
and of holomorphic differentials  associated with a divisor to the
degree of the divisor and the genus of the curve:
\[
\dim H^0_D-\dim H^1_D = 1-g+ \text{deg} D.
\]
Here $H^0_D$ is the $\mathbb C$-vector space of meromorphic functions having poles of order at most $n_i$
at each $P_i$, $H^1_D$ the $\mathbb C$-vector space of holomorphic
differentials with zeros of order at least $n_i$ at each $P_i.$

More useful for our immediate purpose is the formulation,
\[
(\dim H^0_{D+P}-\dim H^0_D)+(\dim H^1_D-\dim H^1_{D+P})=1.
\]

Thus it is that in adding a point $P$ to the divisor $D$,
\emph{either} $H^0$ augments by $+1$ and $H^1$ is unaltered
\emph{or} $H^1$ augments by $-1$ and $H^0$ remains unaltered.

At a generic point, $P,$ on the curve, consider the family of
divisors $nP$ for $n\in {\mathbb N}\cup\{0\}$. The following
pattern of dimensions obtains.
\[
\begin{array}{cccccccc}
%&\bullet&\circ&\cdots&\circ&\bullet&\bullet&\cdots\\
D&0&P&\ldots&gP&(g+1)P&(g+2)P&\dots\\
\dim H^0_D&1&1&\ldots&1&2&3&\ldots\\
\dim H^1_D&g&g-1&\ldots&0&0&0&\ldots\\
\end{array}
\]

Note that there are $g$ `gaps' where the addition of an extra $P$
results in no new functions on $X.$ Each gap is associated with the
loss of a holomorphic differential. The pattern is, however, broken
at special choices of divisor, more specifically special choices of
$P$ in the sequence above. At such \emph{Weierstra{\ss} points} the
interval of $g$ gaps is broken as illustrated in the examples below.

\subsubsection*{Example: Genus two hyperelliptic case}
There are Weierstra{\ss} gaps at $P$ and $3P$:
\[
\begin{array}{ccccccccccccc}
%&\bullet&\circ&\bullet&\circ&\bullet&\bullet&\bullet&\bullet&\bullet&\bullet& \quad \,\, \bullet \,\, \cdots \\
D&0&P&2P&3P&4P&5P&6P&7P&8P&9P&10P\\
\dim H^0_D&1&1&2&2&3&4&5&6&7&8&9\\
\dim H^1_D&2&1&1&0&0&0&0&0&0&0&0\\
&1&              &x&               &x^2&y&x^3&xy&x^4&x^2y&x^5,y^2    \\
& &\downarrow    & &\downarrow     &   & &   &  &   &    &\Downarrow \\
& &\frac{dx}{f_y}& &\frac{xdx}{f_y}&   & &   &  &   &    &f(x,y)
\end{array}
\]

The curve is a relation on $H^0_{10P}$:
\[
f(x,y)=y^2+[x^2]y+[x^5]=0.
\]
The notation $[x^r]$ denotes a polynomial of degree $n$ in $x.$
A relation arises at $10P$ because we can create elements in
$H^0_{10P}\backslash H^0_{9P}$ in two ways when only one extra
dimension is available. The diagram also indicates the form of
holomorphic differential (depending on $f(x,y)$) lost at each gap.

\subsubsection*{Example: Genus three non-hyperelliptic case}
There are three Weierstra{\ss} gaps at $P$, $2P$ and $5P.$
\[
\begin{array}{ccccccccccccccc}
%&\bullet&\circ&\circ&\bullet&\bullet&\circ&\bullet&\bullet&\bullet&\bullet&\bullet&\bullet& \bullet  \\
D&0&P&2P&3P&4P&5P&6P&7P&8P&9P&10P&11P&12P\\
\dim H^0_D&1&1&1&2&3&3&4&5&6&7&8&9&10\\
\dim H^1_D&3&2&1&1&1&0&0&0&0&0&0&0&0\\
&1&&&x&y&&x^2&xy&y^2&x^3&x^2y&xy^2&x^4,y^3\\
&&\downarrow&\downarrow&&&\downarrow&&&&&&&\Downarrow\\
&&\frac{ydx}{f_y}&\frac{dx}{f_y}&&&\frac{xdx}{f_y}&&&&&&&f(x,y)
\end{array}
\]
The curve is a relation on $H^0_{12P}$:
\[
f(x,y)=y^3+[x]y^2+[x^2]y+[x^4]=0.
\]

\subsubsection*{Transformations}
In each of the examples above $x$ and $y$ are coordinates on $X$, i.e. maps $X\rightarrow \mathbb{P}^1.$
We shall be interested in coordinate transformations
\[
x\mapsto \tilde x=\psi\circ x \quad \mbox{for} \quad \psi:\mathbb{P}^1\rightarrow \mathbb{P}^1.
\]
Note that under such transformations the coordinates $x$ and $\tilde
x$ have poles at distinct divisors on $X$: Weierstra{\ss} points for
$x$ map to more general \emph{special} divisors for $\tilde x$ and
give different models of $X$: rationally equivalent, and usually
singular, curves in $\mathbb{P}^2.$ For example, the following three
curves are all rationally equivalent.
\[
y^3+[x^3][x]^2=0\quad \longleftrightarrow\quad \tilde
y^3+[\tilde x^4][\tilde x]^2=0\quad\longleftrightarrow\quad [\hat x]\hat y^3+[\hat x^4]=0
\]
It so happens that this example is of genus three. Every genus three
$X$ has a nonsingular embedding in $\mathbb{P}^2$ and the third
equation above is it, \cite{FarkasKra}.

\section{Working with Kleinian functions} \label{SEC_Sigma}

Weierstra{\ss} introduced an auxiliary function, $\sigma(u)$, in his theory which satisfies
\begin{eqnarray}
\wp(u) &=& - \frac{d^2}{d u^2} \log \big[ \sigma(u) \big].  \label{eq:Intro_elliptic_ps}
\end{eqnarray}
The $\sigma$-function plays a crucial role in both the generalisation and in applications of the theory.  It satisfies the following two term addition formula,
\begin{eqnarray}
- \frac{\sigma(u+v)\sigma(u-v)}{\sigma(u)^2\sigma(v)^2} = \wp(u) - \wp(v). \label{eq:Intro_elliptic_add}
\end{eqnarray}
The Weierstra{\ss} $\sigma$-function is usually defined as an infinite product over the periods, but it can also be expressed using the first Jacobi $\theta$-function, multiplied by a constant and exponential factor.  It is this definition that generalises naturally to give higher genus $\sigma$-functions as multivariate functions defined using the Riemann $\theta$-function, as presented in Definition \ref{def:HG_sigma}.  Once the generalised $\sigma$-function is defined we may then consider generalised $\wp$-functions in analogy to equation (\ref{eq:Intro_elliptic_ps}), as presented in Definition \ref{def:nip}.

This approach was pioneered by Klein and Baker as described in Baker's classic texts \cite{ba97} and \cite{ba07}.  Hence these generalised functions are sometimes called Kleinian.  The periodicity conditions of all the higher genus functions are defined using the period lattice of an underlying algebraic curve.  In this section we work with the following class of curves, while in Section \ref{SEC_Equivar} we consider an equivariant class.

\begin{definition} \label{def:HG_general_curves}
For two coprime integers $(n,s)$ with $s>n$ we define an \textbf{$\bm{(n,s)}$-curve}, denoted C, as an algebraic curve defined by $f(x,y)=0$, where
\begin{equation} \label{eq:general_curve}
f(x,y) = y^n + p_1(x)y^{n-1} + p_{2}(x)y^{n-2} + \cdots + p_{n-1}(x)y - p_{n}(x).
\end{equation}
Here $x$, $y$ are complex variables and $p_j(x)$ are polynomials in $x$ of degree (at most) $\lfloor{ js } / n \rfloor$.  We define a simple subclass of the curves by setting $p_j(x)=0$ for $0\leq j \leq n-1$.  These curves are defined by
\begin{equation} \label{eq:ct_curve}
f(x,y) = y^n - (x^s +\lambda_{s-1}x^{s-1}+\dots+\lambda_{1}x+\lambda_0)
\end{equation}
and are called \textbf{cyclic $\bm{(n,s)}$-curves}.
We denote the curve constants by $\lambda_j$ for the cyclic curves and by $\mu_j$ for the general $(n,s)$-curves.
Note that in the literature the word `cyclic' is sometimes replaced by `strictly' or `purely $n$-gonal'.
\end{definition}
Restricting to the cyclic classes results in much easier computation, but it not usually necessary for theoretical reasons.  However, the cyclic curves do possess extra symmetry which manifests itself in the associated functions satisfying a wider set of addition formulae.

If the curve $C$ is a non-singular model of $X$ then the genus is given by $g=\frac{1}{2}(n-1)(s-1)$ and the associated functions, to be defined shortly, will be multivariate with $g$ variables, $\bm{u} = (u_1, \dots, u_g)$.  As an example, the elliptic curve in equation (\ref{eq:Intro_ec}) is a (2,3)-curve and the associated Weierstra{\ss} $\sigma$ and $\wp$-functions depend upon a single complex variable $u$.

\subsection{Defining the functions}

We first describe how the period lattice associated to the curve may be constructed.  We start by choosing a basis for the space of differential forms of the first kind; the differential 1-forms which are holomorphic on the curve, $C$.  There is a standard procedure to construct this basis for an $(n,s)$-curve, (see for example \cite{N10}).  In general the basis is given by $h_idx/f_y, i=1 \dots g$ where the $h_i$ are monomials in $(x,y)$ whose structure may be predicted by the Weierstra{\ss} gap sequence, although other normalisations are sometimes used.

We next choose a symplectic basis in \(H_1(C,\mathbb{Z})\) of cycles (closed paths) upon the compact Riemann surface defined by $C$.  We denote these by
$\{\alpha_1, \dots \alpha_g,$ $\beta_1, \dots \beta_g\}$.  We ensure the cycles have intersection numbers
\begin{eqnarray*}
\alpha_i \cdot \alpha_j = 0, \qquad \beta_i \cdot \beta_j = 0, \qquad
\alpha_i \cdot \beta_j = \delta_{ij} =
\left\{ \begin{array}{rl}
1 & \mbox{if }  i = j, \\
0 & \mbox{if }  i \neq j.
\end{array} \right.
\end{eqnarray*}
The choice of these cycles is not unique, but the functions will be independent of the choice.

We introduce $\bm{dr}$ as a basis of differentials of the second kind.  These are meromorphic differentials on $C$ which have their only pole at $\infty$.  This basis is usually derived alongside the fundamental differential of the second kind.  Rather than repeat the full details here we refer the reader to \cite{ba97} for the general theory and \cite{bg06} which gives a detailed example construction for the (3,5)-curve.

We can now define the standard period matrices associated to the curve as
\begin{eqnarray*}
\begin{array}{cc}
      \omega'  = \left( \oint_{\alpha_k} du_\ell \right)_{k,\ell = 1,\dots,g} &
\qquad\omega'' = \left( \oint_{ \beta_k} du_\ell \right)_{k,\ell = 1,\dots,g}  \\
        \eta'  = \left( \oint_{\alpha_k} dr_{\ell} \right)_{k,\ell = 1,\dots,g} &
\qquad  \eta'' = \left( \oint_{ \beta_k} dr_{\ell} \right)_{k,\ell = 1,\dots,g}
\end{array}.
\end{eqnarray*}
We define the period lattice $\Lambda$ formed from $\omega', \omega''$ by
\[
\Lambda = \big\{ \omega'\bm{m} + \omega''\bm{n}, \quad \bm{m},\bm{n} \in \mathbb{Z}^g \big\}.
\]
Note the comparison with equation (\ref{eq:Intro_period}) and that the period matrices play the role of the scalar periods in the elliptic case.  The functions we treat are defined upon \(\mathbb{C}^g\) with coordinates usually expressed as
\[
\bm{u}=(u_1, \dots, u_g).
\]
%Note that any point $\bu \in \mathbb{C}^g$ can be expressed as
%\begin{eqnarray*}
%\bu &= \sum_{i=1}^g \int_{\infty}^{P_i} \bm{du},
%\end{eqnarray*}
%where the $P_i$ are points upon $C$.
The period lattice $\Lambda$ is a lattice in the space $\mathbb{C}^g$.  The Jacobian variety of $C$ is presented by $\mathbb{C}^g/\Lambda$, and is denoted by $J$. We define \(\kappa\) as the modulo \(\Lambda\) map,
\[
\kappa \ : \ \mathbb{C}^g \to J.
\]
For $k=1$, $2$, $\dots$ define $\mathfrak{A_k}$, the \emph{Abel map} from the $k$-th symmetric product \(\mathrm{Sym}^k(C)\) of \(C\)  to $J$ by
\begin{eqnarray*}
\mathfrak{A_k}: \mbox{Sym}^k(C) &\to&     J \nonumber \\
(P_1,\dots,P_k)   &\mapsto&
\left( \int_{\infty}^{P_1} \bm{du} + \dots + \int_{\infty}^{P_k} \bm{du} \right) \pmod{\Lambda},
\label{eq:Abel}
\end{eqnarray*}
where the $P_i$ are again points upon $C$.  Denote the image of the $k$-th Abel map by $W^{[k]}$ and define the \emph{$k$-th standard theta subset} (often referred to as the $k$-th strata) by
\begin{eqnarray*}
\Theta^{[k]} = W^{[k]} \cup [-1]W^{[k]},
\end{eqnarray*}
where \([-1]\) means that
\begin{eqnarray*}
[-1](u_1, \dots ,u_g) = (-u_1, \dots ,-u_g).
\end{eqnarray*}

We are considering functions that are periodic with respect to the lattice $\Lambda$.
\begin{definition} \label{def:HG_Abelian}
Let $\mathfrak{M}(\bu)$ be a meromorphic function of $\bu \in \mathbb{C}^g$.  Then $\mathfrak{M}$ is a \textbf{standard Abelian function associated with $\bm{C}$} if it has poles only along \(\kappa^{-1}(\Theta^{[g-1]})\) and satisfies, for all $\bm{\ell}\in \Lambda$,
\begin{equation} \label{eq:HG_Abelian}
\mathfrak{M}(\bu + \bm{\ell}) = \mathfrak{M}(\bu).
\end{equation}
\end{definition}
The generalisations of the Weierstra{\ss} $\wp$-functions with which we deal will satisfy equation (\ref{eq:HG_Abelian}).  The generalisation of the $\sigma$-function we define below will be quasi-periodic.  Let \(\bm{\delta} = \omega'\bm{\delta'}+\omega''\bm{\delta''}\) be the Riemann constant with base point $\infty$.  Then $[\bm{\delta}]$ is the theta characteristic representing the Riemann constant for the curve C with respect to the base point $\infty$ and generators $\{\alpha_j,\ \beta_j\}$ of $H_1(C,\mathbb{Z})$.  (See for example \cite{bel97} pp23-24.)

\begin{definition} \label{def:HG_sigma}
The \textbf{Kleinian $\bm{\sigma}$-function} associated to a general $(n,s)$-curve is defined using a multivariate $\theta$-function with characteristic \(\bm{\delta}\) as
\begin{eqnarray*}
\sigma(\bu) &=& c \exp \big( \textstyle \frac{1}{2} \bm{u} \eta' (\omega')^{-1} \bm{u}^T \big)
\cdot \theta[\bm{\delta}]\big((\omega')^{-1}\bm{u}^T \hspace*{0.05in} \big| \hspace*{0.05in} (\omega')^{-1} \omega''\big).  \\
&=& c \exp \big( \textstyle \frac{1}{2} \bm{u} \eta' (\omega')^{-1} \bm{u}^T \big)
\displaystyle \sum_{\bm{m} \in \mathbb{Z}^g} \exp \bigg[ 2\pi i \big\{
%\\ & &
\textstyle \frac{1}{2} (\bm{m} + \bm{\delta'})^T (\omega')^{-1} \omega''(\bm{m}
+ \bm{\delta'}) + (\bm{m} + \bm{\delta'})^T ((\omega')^{-1} \bm{u}^T + \bm{\delta''} )\big\} \bigg].
\end{eqnarray*}
The constant $c$ is dependent upon the curve parameters and the basis of cycles and is fixed later,
following Lemma {\rm \ref{lem:sigexp}}.
\end{definition}

We now summarise the key properties of the $\sigma$-function.  See \cite{bel97} or \cite{N10} for the construction of the $\sigma$-function to satisfy these properties.  For any point $\bu \in \mathbb{C}^g$ we denote by $\bu'$ and $\bu''$ the vectors in $\mathbb{R}^g$ such that
\[
\bu=\omega'\bm{u}'+\omega''\bm{u}''.
\]
Therefore a point  $\bm{\ell}\in\Lambda$ is written as
\[
\bm{\ell} = \omega'\bm{\ell'} + \omega''\bm{\ell''} \in \Lambda, \qquad \bm{\ell'},\bm{\ell''} \in \mathbb{Z}^g.
\]
For $\bu, \bv \in \mathbb{C}^g$ and $\bm{\ell} \in \Lambda$, define $L(\bu,\bv)$ and $\chi(\bm{\ell})$ as follows:
\begin{eqnarray*}
L(\bu,\bv) &=& \bu^T \big( \eta'\bm{v'} + \eta''\bm{v''} \big), \\
\chi(\bm{\ell}) &=& \exp \big[ 2 \pi \mbox{i} \big\{  (\bm{\ell'})^T\delta'' - (\bm{\ell''})^T\delta'
+ \textstyle \frac{1}{2}(\bm{\ell'})^T \bm{\ell''} \big\} \big].
\end{eqnarray*}

\begin{lemma}
Consider the $\sigma$-function associated to an $(n,s)$-curve.
\begin{itemize}
\item It is an entire function on $\mathbb{C}^g$.
\item It has zeros of order one along the set  $\kappa^{-1}(\Theta^{[g-1]})$.  Further, we have $\sigma(\bu)~\neq~0$ outside the set.
\item For all $\bu \in \mathbb{C}^g, \bm{\ell} \in \Lambda$ the function has the quasi-periodicity property:
\begin{eqnarray*}
\sigma(\bu + \bm{\ell}) = \chi(\bm{\ell})
\exp \left[ L \left( \bu + \frac{\bm{\ell}}{2}, \bm{\ell} \right) \right] \sigma(\bu).
\end{eqnarray*}
\item It has definite parity given by
\[
\sigma(-\bu) = (-1)^{\frac{1}{24}(n^2-1)(s^2-1)}\sigma(\bu).
\]
\end{itemize}
\end{lemma}

\begin{proof}
The function is clearly entire from the definition, while the zeros and the quasi-periodicity are classical results, (see \cite{ba97}), that are fundamental to the definition of the function.  They both follow from the properties of the multivariate $\theta$-function.  The parity property is given by Proposition 4(iv) in \cite{N10}. \\
\end{proof}
We can now finally define $\wp$-functions using an analogy of equation (\ref{eq:Intro_elliptic_ps}).  Since there is more than one variable we need to be clear which we differentiate with respect to.  We define multiple $\wp$-functions and use the following \emph{index notation}.

\begin{definition} \label{def:nip}
Define \textbf{$\bm{m}$-index Kleinian $\bm{\wp}$-functions} for $m\geq2$ by
\begin{eqnarray*}
\wp_{i_1,i_2,\dots,i_m}(\bu) = - \frac{\partial}{\partial u_{i_1}} \frac{\partial}{\partial u_{i_2}}\dots
\frac{\partial}{\partial u_{i_m}} \log \big[ \sigma(\bu) \big],
%\quad i_1 \leq \dots \leq i_m \in \{1,\dots,g\}.
\end{eqnarray*}
where $i_1 \leq \dots \leq i_m \in \{1,\dots,g\}$.
\end{definition}
The $m$-index $\wp$-functions are meromorphic with poles of order $m$ when $\sigma(\bu)=0$.  We can check that they satisfy equation (\ref{eq:HG_Abelian}) and hence they are Abelian.   The $m$-index $\wp$-functions have definite parity with respect to the change of variables $\bu \to [-1]\bu$.  This is independent of the underlying curve, with the functions odd if $m$ is odd and even if $m$ is even.  Note that the ordering of the indices is irrelevant and so for simplicity we always order in ascending value.

\subsection{Weights and expansions}

Klein considered generalisations of the $\wp$-functions associated with certain \emph{hyperelliptic curves}.  These can be thought of as the $(n,s)$-curves with $n=2$, (when $s=3$ the hyperelliptic curves reduce to elliptic curves).  The $(n,s)$-curves where $n=3$ are labeled \textit{trigonal curves} and those with $n=4$ are labeled  \textit{tetragonal curves}.  The original work of Klein and Baker motivated the general definitions above, which were first developed in \cite{bel97} and \cite{eel00}.  Many results for hyperelliptic curves were presented in \cite{bel97} with these methods applied to trigonal curves by the same authors in \cite{bel00}.  The tetragonal curves were first considered in \cite{MEe09}.

The methods of Klein become harder to implement when $n=3$ and harder still when $n=4$.  The difficulties are not just computational but theoretical.  This led to the development of a related approach making use of the series expansion of the $\sigma$-function.  This approach was applied to the two canonical trigonal cases, first in \cite{eemop07} and \cite{bego08} and then again after further development in \cite{MEeo11} and \cite{ME11}.  It was essential to the development of the tetragonal theory in \cite{MEe09} and the theory of higher genus trigonal curves in \cite{MEhgt10}.  In this section we discuss how such a series expansion may be derived and used.

\subsubsection*{The Sato weights}

For a given $(n,s)$-curve we can define a set of weights, denoted by $\mathrm{wt}$ and often referred to as the \emph{Sato weights}.  We start by setting $\mathrm{wt}(x)=-n, \mathrm{wt}(y)=-s$ and then choose the weights of the curve parameters to be such that the curve equation is homogeneous.  We see that for cyclic curves this imposes $\mathrm{wt}( \lambda_j )=-n(s-j)$ while for the non-cyclic curves we usually label the $\mu_j$ with their weight.
For example, the general (2,5)-curve is given by
\[
y^2 + y(\mu_{1}x^2 + \mu_{3}x + \mu_{5}) = x^5 + \mu_{2}x^4 + \mu_{4}x^3 + \mu_{6}x^2 + \mu_{8}x + \mu_{10},
\]
and the cyclic restriction by
\[
y^2 = x^5 + \lambda_4x^4 + \lambda_3x^3 + \lambda_2x^2 + \lambda_1x + \lambda_0.
\]
In both cases the weight of $x$ and $y$ are respectively $2$ and $5$.  Both equations are homogeneous of weight $-10$ and so we see that the coefficients $\mu_{j}$ are labeled with the negative of their weight while the weights of the $\lambda_0,\lambda_1,\lambda_2,\lambda_3,\lambda_4$ are given by $10,8,6,4,2$ respectively.

The Abel map $\mathfrak{A}_1$ gives an embedding of the curve $C$ upon which we can define $\xi= x^{-\frac{1}{n}}$ as the local parameter at the origin, $\mathfrak{A}_1(\infty)$.  We can then express $y$ and the basis of differentials using $\xi$ and integrate to give series expansions for $\bu$.  We can check the weights of $\bu$ from these expansions and see that they are prescribed by the Weierstra{\ss} gap sequence.  For example, in the (2,5)-case $\bu=(u_1,u_2)$ has weight $(2,1)$.

By considering Definition \ref{def:nip} we see that the weight of the $\wp$-functions is the negative of the sum of the weights of the variables indicated by the indices, irrespective of what the weight of $\sigma(\bu)$ is.  We note that curves of the same genus will, notationally, have the same $\wp$-functions, but may exhibit different behavior as indicated by the different weights of the variables and functions.  We will discuss the weight of the $\sigma$-function below.  All other functions discussed are constructed from $\sigma$ or $\wp$-functions and their weights follow accordingly.  We can show that all the equations in the theory are homogeneous in these weights, with a more detailed discussion of this available, for example, in \cite{MEe09}.

\subsubsection*{The sigma-function expansion}

We can construct a series expansion of the $\sigma$-function about the origin, as described below.

\begin{lemma} \label{lem:sigexp}
The Taylor series expansion of $\sigma(\bu)$ about the origin may be written as
\[
\sigma(\bu) = K \cdot SW_{n,s}(\bu) + \sum_{k=0}^{\infty} C_{k}(\bu).
\]
Here $K$ is a constant, $SW_{n,s}$ the Schur-Weierstra{\ss} polynomial generated by $(n,s)$ and each $C_k$ a finite, polynomial composed of products of monomials in $\bm{u}$ of weight $k$ multiplied by monomials in the curve parameters of weight $-(\mbox{wt}(\sigma)-k)$.
\end{lemma}

\begin{proof}
We refer the reader to \cite{N10} for a proof of the relationship between the $\sigma$-function and the Schur-Weierstra{\ss} polynomials and note that this was first discussed in \cite{bel99}.  We see that the remainder of the expansion must depend on the curve parameters and split it up into the different $C_k$ according to the weight split on terms between $\bu$ and either the $\mu_i$ or the $\lambda_i$.  We can see that each $C_k$ is finite since the number of possible terms with the prescribed weight properties is finite.  In fact, by considering the possible weights of the curve coefficients we see that the index $k$ in the sum will actually increase in multiples of $n$.\\
\end{proof}

The Schur-Weierstra{\ss} polynomials are Schur polynomials generated by a Weierstra{\ss} partition, derived in turn from the Weierstra{\ss} gap sequence for $(n,s)$.  See \cite{N10} and \cite{bel99} for more details on these polynomials.  This connection with the Schur-Weierstra{\ss} polynomials allows us to determine the weight of the $\sigma$-function as
\[
\mbox{wt}(\sigma) = (1/24)(n^2-1)(s^2-1).
\]
In Definition \ref{def:HG_sigma} we fix $c$ to be the value that makes the constant $K=1$ in the above lemma.  Some other authors working in this area may use a different constant and in general these choices are not equivalent.  However, the constant can be seen to cancel in the definition of all Abelian functions, leaving results between the functions independent of $c$.  Note that this choice of $c$ ensures that the Kleinian $\sigma$-function matches the Weierstra{\ss} $\sigma$-function when the $(n,s)$-curve is chosen to be the classic elliptic curve.
We note that there are other definitions of the $\sigma$-function and a definition more in line with the equivariant approach discussed in Section \ref{SEC_Equivar} is currently a topic of research.

\bigskip

The expansion can be constructed by considering each $C_k$ in turn, identifying the possible terms, forming a series with unidentified coefficients and then determining the coefficients by ensuring the expansions satisfied known properties of the $\sigma$-function.  For example, ensuring it vanishes on the strata or insuring that the known identities between the $\wp$-functions are satisfied.

Large expansions of this type were first introduced in \cite{bg06}, in which used the generalised $\sigma$-function to construct explicit reductions of the Benney equations, (see also \cite{bg03}, \cite{bg04} and
\cite{MEg09}).  Since then they have been an integral tool in the investigation of Abelian functions.  Recently computational techniques based on the weight properties have been used to derive much larger expansions and we refer the reader to \cite{MEe09} and \cite{MEhgt10} for a more detailed discussion of the constructions.

We note that such expansions are possible for the general $(n,s)$-curves, but that the calculations involved are far simpler for the cyclic cases.  However, even in these cases they grow quickly in size and become even larger as the genus grows.  For example, in \cite{MEe09} the construction of the $\sigma$-expansion for the cyclic (4,5)-curve was described in detail.  Here the $\sigma$-function had weight 15 and the expansion was calculated up to $C_{59}$ which had 81,832 non-zero terms, (from 120,964 possible terms with the correct weight structure).

The number of $C_k$ required will depend on how the expansion is to be used.  If used to identify relations as discussed in Section \ref{SEC_DerRes}, then it will need to contain information on the monomials in curve coefficients that may be present in the identity in question.  This can be determined by the weight and type of the desired identity.

\subsection{Bases of Abelian functions} \label{SEC:Bases}

We can classify the Abelian functions according to their pole structure.  We denote by $\Gamma(m)$ the vector space of Abelian functions defined upon $J$ which have poles of order at most $m$, occurring only on the $\Theta$-divisor; the strata $\Theta^{[g-1]}$ where the $\theta$ and $\sigma$-functions have their zeros and the Abelian functions their poles.

A key problem is the generation of bases for these vector spaces.  Note that the dimension of the space $\Gamma(m)$ is $m^g$ by the Riemann-Roch theorem for Abelian varieties, (see for example \cite{la82}).  The first step in constructing a basis is to include the entries in the subsequent basis for $\Gamma(m-1)$. Subsequently only functions with poles of order exactly $m$ need to be sought.  The Kleinian $\wp$-functions are natural candidates and indeed are sufficient to solve the problem in the elliptic case.  Here $g=1$ and so only one new function is required at each stage, which can be filled by the repeated derivatives of the Weierstra{\ss} $\wp$-functions.

However, if $g>1$ then new classes of functions are required to complete the bases.  We discuss the genus two case as it illuminates the connections between the two approaches to $\wp$-functions discussed in this paper.  Table \ref{tab:g2Bases} shows how bases in this case may be constructed for arbitrary $m$.  As in the genus one case the basis for $\Gamma(2)$ may be formed using a constant and the 2-index $\wp$-functions.  However, when considering $\Gamma(3)$ we find that even after including the 3-index $\wp$-functions we still need an additional function.  The function
\[
\Xi=\wp_{11}\wp_{22} - \wp_{12}^2
\]
is usually taken to fill this hole.  The individual terms in $\Xi$ have poles of order 4, but when taken together they cancel to leave poles of order 3 as can easily be checked using Definition \ref{def:nip}.  We then proceed to consider $\Gamma(4)$ and find that two new functions are needed after the inclusion of the 4-index $\wp$-functions.  The two derivatives of $\Xi$ can play this role and we use the notation
\[
\partial_i \Xi = \frac{\partial}{\partial u_i}\Xi
\]
for simplicity, (and similarly for other functions).  Naturally the derivatives of the $\Xi$ have poles of order four but we must also ensure that they are linearly independent of the other functions.  This can be checked trivially by noting that they are of a different weight to both each other and all the other elements of the basis.

It is simple to check that when considering higher values of $m$ we find that the preceding basis and its unique derivatives always give the required number of functions for a new basis, and that they are all of unique weight and so linearly independent.  Hence the general basis is as described in Table \ref{tab:g2Bases}, where $\{ \cdot \}$ is indicating all functions of this form.

\begin{table}
\begin{center}
\begin{tabular}{l | c l}
\textbf{Space} \,& \,\textbf{Dimension} & \textbf{Basis} \\ \hline
$\Gamma(0)$      & 1          & $\{1\}$                                                                                 \\
$\Gamma(1)$      & 1          & $\{1\}$                                                                                 \\
$\Gamma(2)$      & 4          & $\{1,\wp_{11},\wp_{12},\wp_{22}\}$                                                     \\
$\Gamma(3)$      & 9          & $\{1,\wp_{11},\wp_{12},\wp_{22}, \wp_{111},\wp_{112}, \wp_{122}, \wp_{222}, \Xi \}$  \\
$\Gamma(4)$      & 16         & $\{1,\dots,\Xi,\wp_{1111},\dots,\wp_{2222}, \partial_1\Xi, \partial_2\Xi \}$   \\
$\vdots$         & $\vdots$   & \\
$\Gamma(m)$      & $m^2$      & $\{ \mbox{basis for } \Gamma(m-1)\}\cup\{\{\wp_{i_1 \dots i_m}\}, \{\partial_{i_1}\dots\partial_{i_{m-2}}\Xi\}\}$ \\
\end{tabular}
\end{center}
\caption{Table of bases for Abelian functions associated with a genus two curve.}
\label{tab:g2Bases}
\end{table}

The genus 1 and 2 cases, in which we get to a stage where new bases are calculated from the old ones and their derivatives, are special.  They fall into the class where the theta divisor is non-singular and the $\mathcal{D}$-module structure of such cases is discussed in \cite{cn08}.  However, subsequent $(n,s)$-curves are not covered by this case and so new methods must be used to derive bases here.  Two approaches have been developed to find additional functions.  (Note that in each case we still need to check the linear independence of the functions, which can be done using the $\sigma$-expansion.)

The first is to match poles in algebraic combinations of $\wp$-functions so that they cancel, analogous to $\Xi \in \Gamma(3)$.  For example, consider the following function, constructed of products of 2 and 3-index $\wp$-functions.
\begin{align*}
\mathcal{B}_{ijklm} &= \wp_{ij}\wp_{klm} + \textstyle \frac{1}{3}\big( \wp_{jk}\wp_{ilm} + \wp_{jl}\wp_{ikm} + \wp_{jm}\wp_{ikl}
%\\ &\qquad
- 2\wp_{kl}\wp_{ijm} - 2\wp_{km}\wp_{ijl} - 2\wp_{lm}\wp_{ijk} \big).
\end{align*}
Each term here has poles of order five, but overall both these and the poles of order four cancel (as can easily be checked using Definition \ref{def:nip}) and so the function belongs to $\Gamma(3)$.  This particular function can be used to construct classes of differential equations bilinear in the 2 and 3-index $\wp$-functions, discussed in more detail in \cite{ME11}.  This paper also constructs other classes of such combinations in a systematic way and uses them in an example case to find new addition formulae and differential equations.  (See also \cite{MEeo11} for more results that followed from this approach.)

The second approach is to define Abelian functions by the application of differential operators on $\sigma$-functions.  For example, define Hirota's bilinear operator as $\mathcal{D}_i = \partial / \partial u_i - \partial / \partial v_i$.  Then define the $m$-index $Q$-functions (for $m$ even) as
\begin{align*}
Q_{i_1, i_2,\dots,i_m}(\bm{u}) &=  \frac{(-1)}{2\sigma(\bm{u})^2} \mathcal{D}_{i_1}\mathcal{D}_{i_2}...\mathcal{D}_{i_m} \sigma(\bm{u}) \sigma(\bm{v})
\hspace*{0.08in} \Big|_{\bm{v}=\bm{u}}
%\\ &
\qquad \qquad i_1 \leq ... \leq i_n \in \{1,\dots,g\}.
\end{align*}
These are all Abelian functions with poles of order two and so can be used to solve the basis problem in general for $\Gamma(2)$.

Both these approaches construct new functions that have a given pole structure in general, that is without reference to a specific underlying curve.  However, there are some interesting examples of functions which have poles of lower order on certain curves.  Consider the following $\Delta$-function, originally introduced by Baker in \cite{ba03}.
\[
\Delta=\wp_{11}\wp_{33}-\wp_{12}\wp_{23}-\wp_{13}^2+\wp_{13}\wp_{22}.
\]
The terms in this expression have poles of order four and overall the expression will have poles of order three.  (This may be checked using Definition \ref{def:nip}).  However, in the case of the (2,7)-curve, these can be shown to cancel to leave poles of order two.  The (2,7)-case has $g=3$ and so $2^3=8$ functions are required to form a basis for $\Gamma(2)$.  Such a basis may be formed from a constant, the six 2-index $\wp$-functions and the $\Delta$-function.  Similar functions (linear combinations of quadratics in 2-index $\wp$-functions) were constructed in \cite{MEeo11} to complete the basis of $\Gamma(2)$ in the $(2,9)$-case and it appears that such functions are a feature of the hyperelliptic cases.  Their appearance is more natural when considering the equivariant approach as discussed in Section \ref{SEC_Equivar}.

\subsection{Deriving differential equations and addition formulae} \label{SEC_DerRes}

The $\sigma$-expansion and bases of functions discussed above may be used to derive various identities and properties of the functions.  For every $(n,s)$-curve there exist sets of differential equations that generalise the elliptic identities (\ref{eq:Intro_elliptic_diff1}) and (\ref{eq:Intro_elliptic_diff2}).  However, the exact form of these generalisations is not clear.  In the hyperelliptic cases they have an analogous form, so that the generalisation of
\[
\wp''(u) = 6\wp(u)^2 - \textstyle \frac{1}{2}g_2,
\]
is a set of equations expressing each 4-index $\wp$-function as a quadratic equation in 2-index $\wp$-functions.  Similarly the generalisation of
\[
\big(\wp'(u)\big)^2 = 4\wp(u)^3 - g_2\wp(u) - g_3,
\]
is a set of equations expressing products of 3-index $\wp$-function as a cubic equation in 2-index $\wp$-functions.  However, in non-hyperelliptic cases it is not possible to find complete sets, or more accurately, we need to use all the functions in the basis of $\Gamma(2)$ to construct the equations, rather than just the 2-index $\wp$-functions.  These and other sets of differential equations between the $\wp$-functions are discussed in detail in \cite{ME11}.  It is expected that the precise structure of these equations will become apparent following development of the equivariant approach.

In low genus cases such differential equations follow from an algebraic definition of the $\wp$-functions as described in \cite{bel97}.  However, for higher genus cases the $\sigma$-expansion is necessary to find an exhaustive set.  Individual relations can be calculated by forming the relations with unidentified coefficients, based on the weight and structure of the relations.  The coefficients may then be identified by substituting in the $\sigma$-expansion, setting the coefficients to zero and solving the overdetermined resulting linear system.  The procedure to do this is discussed in more detail in \cite{MEe09} and \cite{MEhgt10}.

A similar use of the $\sigma$-expansion is in the construction of addition formulae.  The Weierstra{\ss} formula in equation (\ref{eq:Intro_elliptic_add}) has been generalised in many cases and papers.  The simplest generalisation is to a formula of the form
\[
\frac{\sigma(\bm{u} + \bm{v})\sigma(\bm{u} - \bm{v})}{\sigma(\bm{u})^2\sigma(\bm{v})^2} = \sum_i c_i A_i(\bm{u})B_i(\bm{v})
\]
where $A_i, B_i \in \Gamma(2)$ and the $c_i$ are constants.

This structure follows from linear algebra after checking that the left hand side is both Abelian in $\bu$ and $\bv$.  The coefficients $c_i$ can be identified by forming the right hand side from elements of $\Gamma(2)$ and arbitrary coefficients and then substituting in the $\sigma$-expansion.  There are many simplifications to be made to this calculation by considering the weight, parity and symmetry properties of the left hand side and constructing the right hand side accordingly.  Further computational simplifications can be made by expanding the resulting products of series in a way that takes into account the homogeneity of the weight structure.

There are a second class of addition formulae present only when the underlying $(n,s)$-curve is cyclic,
\[
y^{n} = x^{s} + \lambda_{s-1}x^{s-1} + ... + \lambda_1x + \lambda_0.
\]
It is associated with their extra family of automorphisms:
\[
[\zeta^j]: (x,y) \rightarrow (x, \zeta^j y), \qquad \mbox{where } \zeta^n = 1,
\]
with the following function becoming Abelian in the $n$ variables $\bm{u}^{[i]}, i=1 \dots n$.
\[
\prod_{j=1}^{n} \frac{\sigma\left( \sum_{i=1}^n [\zeta^{i+j}]\bm{u}^{[i]} \right) }{\sigma( (\bm{u}^{[j]})^n ) }
\]
It can hence be written as a sum of products of $n$ functions from $\Gamma(n)$ in the variables, $\bm{u}^{[i]}, i=1 \dots n$ respectively.

Such formulae can be computationally difficult to compute.  Simplified versions may be found, setting one or more of the variables $\bm{u}^{[i]}$ to zero.
We can also consider reduced curves which have further automorphisms and hence extra addition formulae.  For example if we set all the curve coefficients except $\lambda_{0}$ to zero then there will be automorphisms
\[
[\eta^j]: (x,y) \mapsto (\eta^j x,y), \qquad \mbox{where } \eta^s = 1.
\]

The first example of these automorphism addition formulae was given in \cite{eemop07} with further examples recently published in \cite{MEeo11} and \cite{ME11}.

\section{The equivariant approach} \label{SEC_Equivar}

\subsection{Equivariant curves and functions }

The motivations behind the equivariant theory are to shortcut some
of the calculational intricacies of the $\sigma$-expansion and to
keep a firm grip of the underlying geometry of the curve.
This approach was the topic of \cite{CA08} where it was developed for an equivariant hyperelliptic curve of genus $g$, which is given by
\[
y^2 = \sum_{i=1}^{2g+2} \binom{2g+2}{i} \lambda_ix^i.
\]
We start by choosing simple $PSL(2,\C)$ coordinate transformations on the curve coordinates
\[
x\rightarrow\tilde x = \frac{\alpha x+\beta}{\gamma x+\delta}, \qquad y\rightarrow\tilde y = \frac{y}{(\gamma x+\delta)^p}.
\]
Then $H^1_0$, (as defined in Section \ref{SubSec_RR}) is a $g$-dimensional $SL(2,\C)$ module which may be reducible.
It should be noted that this choice of transformation does not exhaust all possible and possibly useful transformations.
They provide a convenient and manageable starting point.

We define functions $\wp_{ij}$ associated with a hyperelliptic curve of genus $g$ algebraically as those satisfying
\[
\frac{1}{\delta^{g-1}}\sum_{i,j=1}^{g}x^{i-1}x_k^{j-i}\wp_{ij}=\frac{yy_k-F(x,x_k)}{\delta^{g+1}}
\]
Here the quantity $I=yy_k-F(x,x_k)$ is an `equivariant polar form' of the curve.  For example, if the genus is three then
\begin{align*}
F(x,x_k) &= \lambda_8(xx_k)^4 + 4\lambda_7(xx_k)^3(x+x_k) + \lambda_6(xx_k)^2(3x^2+8xx_k+3x_k^2)
+ 4\lambda_5(xx_k)(x^3+6x^2x_k+6xx_k^2+x_k^3) \\
&\qquad + 4\lambda_4(x^4+16x^3x_k+36x^2x_k^2+16xx_k^3+x_k^4)
+ 4\lambda_3(x^3+6x^2x_k+6xx_k^2+x_k^3) \\
&\qquad + \lambda_2(3x^2+8xx_k+3x_k^2) + 4\lambda_1(x+x_k) + \lambda_0.
\end{align*}
This definition guarantees the equivariance of the algebraic definition at the expense of some small changes (addition of constants) to the classical definition of the $\wp_{ij}.$

We denote the standard, finite dimensional irreducible $SL(2,\C)$-module of dimension $d$ as $\mathbf d$, and use similar notation for other modules.
In the hyperelliptic $(2,2g+1)$ cases (of genus $g$), $H^1_0$ is a $\mathbf g$ and the $\wp_{ij}$ decompose as
\[
{\mathbf g}\odot{\mathbf g} = (\mathbf{2g-1})\oplus(\mathbf{2g-5})\oplus\ldots
\]
For the trigonal $(3,3p+1)$ cases (of genus $3p$), $H^1_0$ is a ${\mathbf p}\oplus{\mathbf{2p}}$ and the $\wp_{ij}$ decompose as
\begin{align*}
{\mathbf p}\odot{\mathbf p} &= (\mathbf {2p-1})\oplus(\mathbf {2p-5})\oplus\ldots
\\
{\mathbf p}\otimes\mathbf{2p} &= ( \mathbf {3p-1} )\oplus(\mathbf {3p-3})\oplus(\mathbf {3p-5})\ldots
\\
\mathbf {2p}\odot\mathbf {2p} &= (\mathbf {4p-1})\oplus(\mathbf {4p-5})\oplus\ldots
\end{align*}

\subsection{ Identities} \label{SEC:Equi_Ident}

Products of $\wp$-functions are to be thought of as symmetric tensor
products of finite dimensional modules. Then any identity must
belong to some finite dimensional, irreducible module of identities
classified by some highest weight element. The strategy then is to:
\begin{itemize}
\item start from the Kleinian, algebraic definition of $\wp_{ij};$
\item undertake a (classical) singularity expansion, analogous to the classical elliptic methods, to obtain some highest \emph{weight} identity;
\item generate the module of identities associated with this highest weight identity;
\item execute a further singularity expansion modulo this module to obtain further identities.
\end{itemize}

The tools are provided by the representation theory: the use of the
Lie algebra $\mathfrak{sl}(2,\C)$ to locate highest weight elements
and the \emph{Casimir} element of the enveloping algebra to identify
dimensions.

\subsubsection*{Example: Explicit Analogue to the Weierstra{\ss} cubic in the genus three hyperelliptic case}

Define the following matrices:
\[
P = \left[
\begin{array}{ccccc}
0 & 0 & \wp_{11} & \wp_{12} & \wp_{13} \\
%&&&&\\
0 & -2\wp_{11} & -\wp_{12} & \wp_{22}-2\wp_{13} & \wp_{23}\\
%&&&&\\
\wp_{11} & -\wp_{12} & 2\wp_{13}-2\wp_{22} & -\wp_{23} & \wp_{33}\\
%&&&&\\
\wp_{12} & \wp_{22}-2\wp_{13} & -\wp_{23} & -2\wp_{33} & 0\\
%&&&&\\
\wp_{13} & \wp_{23} & \wp_{33} & 0 & 0
\end{array}
\right],
\qquad \qquad
H=\left[
\begin{array}{ccccc}
\lambda_0  & 4\lambda_1  & 6\lambda_2  & 4\lambda_3  & \lambda_4   \\
4\lambda_1 & 16\lambda_2 & 24\lambda_3 & 16\lambda_4 & 4\lambda_5  \\
6\lambda_2 & 24\lambda_3 & 36\lambda_4 & 24\lambda_5 & 6\lambda_6  \\
4\lambda_3 & 16\lambda_4 & 24\lambda_5 & 16\lambda_6 & 4\lambda_7  \\
\lambda_4  & 4\lambda_5  & 6\lambda_6  & 4\lambda_7  & \lambda_8
\end{array}
\right]
\]
\[
A=\left[
\begin{array}{ccccc}
0 & -\wp_{333} & \wp_{233} & -\wp_{223}+\wp_{133} & \wp_{222}-2\wp_{123} \\
%&&&&\\
\wp_{333} & 0 & -\wp_{133} & \wp_{123} & -\wp_{122}+\wp_{113}\\
%&&&&\\
-\wp_{233} & \wp_{133} & 0 & -\wp_{113} & \wp_{112}\\
%&&&&\\
\wp_{223}-\wp_{133} & -\wp_{123} & \wp_{113} & 0 & -\wp_{111}\\
%&&&&\\
-\wp_{222}+2\wp_{123} & \wp_{122}-\wp_{113} & -\wp_{112} & \wp_{111} & 0
\end{array}
\right].
\]
Then all identities of degree two in 3-index $\wp$-functions and degree three in 2-index $\wp$-functions are contained within,
\begin{equation}\label{eq:g3equi_QR}
({\bf l}^TA{\bf k})({\bf l}'^TA{\bf k}')+\frac14\det \left[
\begin{array}{ccc}
H-2P & {\bf l}^T & {\bf k}^T\\
{\bf l'} & 0 & 0\\
{\bf k'} & 0 & 0
\end{array}\right]=0
\end{equation}
where the $\bm{l},\bm{k},\bm{l'},\bm{k'}$ are vectors of dimension 5 with arbitrary parameters, \cite{CA11}. \\

For comparison, the genus one equation is written below in a corresponding fashion.  (Note that here $\wp_{111}\equiv \wp'$).
\begin{equation}
\wp_{111}^2+\frac14\det\left[
\begin{array}{ccc}
\lambda_0   & 2\lambda_1           & \lambda_2-2\wp_{11}\\
2\lambda_1  & 4\lambda_2+4\wp_{11} & 2\lambda_3\\
\lambda_2-2\wp_{11} & 2\lambda_3   & \lambda_4
\end{array}
\right]=0. \nonumber
\end{equation}

Regarding the genus $3$ case we make some remarks.
Firstly, the matrix $A$ is of rank two. As a $5\times 5$ anti-symmetric matrix this implies `Pl{\"u}cker-like' relations on the $\wp_{ijk}$.
This correspondence between the entries of $A$ and Pl{\"u}cker coordinates on the Grassmanian G(2,5) can be made precise and as an expression of Koszul duality which, in turn, suggests a natural conjecture for higher genus identities in the hyperelliptic cases.

Secondly, equivariant non-hyperelliptic cases are still under
construction. The known formulae are certainly equivariant at
leading order in the $\wp_{i_1i_2 \ldots i_n}$ but not being derived from
fully equivariant models of $X$ the devil is in the lower order
detail.

\section{Combining the approaches} \label{SEC_Future}

We aim to fuse the two approaches. We already know how to map from known hyperelliptic equivariant results to the standard results for $(n,s)$-curves.  For example, consider the equivariant theory for the hyperelliptic genus 3 curve discussed in Section \ref{SEC:Equi_Ident}.  The set of identities encoded in the matrix equation (\ref{eq:g3equi_QR}) is equivalent to the set derived and presented for the cyclic (2,7)-curve in Theorem 4.1 of \cite{MEeo11}.  To map from the former to the latter we need to shift the 2-index $\wp$-functions as follows:
\begin{align*}
&\wp_{1,1} \mapsto \wp_{1,1}+3\lambda_2, \quad
\wp_{1,2} \mapsto \wp_{1,2}+2\lambda_3,  \quad
\wp_{1,3} \mapsto \wp_{1,3}+\tfrac{1}{2}\lambda_4,  \quad \\
&\wp_{2,2} \mapsto \wp_{2,2}+9\lambda_4,  \quad
\wp_{2,3} \mapsto \wp_{2,3}+2\lambda_5,  \quad
\wp_{3,3} \mapsto \wp_{3,3}+3\lambda_6. \quad
\end{align*}
Then we need to change the scaling of the curve coefficients:
\begin{align*}
\begin{array}{lllll}
  \lambda_0 \mapsto 4\lambda_0, \quad
& \lambda_1 \mapsto \frac{1}{2}\lambda_1, \quad
& \lambda_2 \mapsto \frac{1}{7}\lambda_2, \quad
& \lambda_3 \mapsto \frac{1}{14}\lambda_3, \quad
& \lambda_4 \mapsto \frac{2}{35}\lambda_4, \quad \\
\lambda_5 \mapsto \frac{1}{14}\lambda_5, \quad
& \lambda_6 \mapsto \frac{1}{7}\lambda_6, \quad
& \lambda_7 \mapsto \frac{1}{2}, \quad
& \lambda_8 \mapsto 0.
&
\end{array}
\end{align*}
Since the 2-index $\wp$-function only shift by a constant we find that the higher index $\wp$-functions are identical in both cases.

The equivariant technology can help make sense of the structure of
the pole bases of the functions discussed in Section
\ref{SEC:Bases}.  Consider again Table \ref{tab:g2Bases} which
summarised the bases in the genus two case.  The starting point was
the three 2-index $\wp$-functions, which formed a 3-dimensional
representation in the equivariant theory.
\[
\bm{3} = \{\wp_{11}, \wp_{12}, \wp_{22} \}.
\]
If we consider tensoring this with itself then we find
\[
\bm{3} \odot  \bm{3} = \bm{5} \oplus \bm{1},
\]
so we have a 5-dimensional representation and a 1-dimensional
representation of functions quadratic in the 2-index
$\wp$-functions.  We have
\begin{align*}
\bm{1} &= \Xi \\
&= \wp_{11}\wp_{22} - \wp_{12}^2 \in \Gamma(3).
\end{align*}
The other representation is
\[
\bm{5} = \{\wp_{11}^2,4\wp_{11}\wp_{12},4\wp_{12}^2+2\wp_{11}\wp_{22},4\wp_{12}\wp_{22},\wp_{22}^2\}
\]
which are all linearly independent functions with poles of order 4 and so can replace the 4-index $\wp$-functions in $\Gamma(4)$.  So tensoring up the  representations allows us to move through the bases and see how they break up into representations themselves, often in ways not apparent from the notation used.

We can also gain insight and move between the representations by considering the representation of operators
\[
\bm{\partial} = \{ \partial_1, \partial_2 \}.
\]
Since this is two dimensional we will have
\[
\bm{\partial} \otimes \bm{n} = \bm{n-1} \oplus \bm{n+1}
\]
when it is tensored with any representation of functions.  For example, tensoring it with the $\bm{3}$ from above gives
\[
\bm{\partial} \odot \bm{3} = \bm{2} \oplus \bm{4}
\]
where $\bm{4}$ is the representation of 3-index $\wp$-functions and $\bm{2}$ is empty due to the integrability properties:
\[
\partial_2\wp_{11}-\partial_1\wp_{12}=0 \qquad \mbox{and} \qquad \partial_1\wp_{22}-\partial_2\wp_{12}=0.
\]
These kinds of observation will lead us to update Table \ref{tab:g2Bases} to give Table \ref{tab:g2EquiBases}.  We have used the notation $\pi_n$ to denote projection onto an $n$-dimensional irreducible submodule.

\begin{table}
\begin{center}
\begin{tabular}{l | c l}
\textbf{Space} \,& \,\textbf{Dimension} & \textbf{Basis} \\ \hline
$\Gamma(0)$      & 1          & $\{1\}$                                                                                 \\
$\Gamma(1)$      & 1          & $\{1\}$                                                                                 \\
$\Gamma(2)$      & 4          & $\{1,{\bf 3\}}$                                                     \\
$\Gamma(3)$      & 9          & $\{1,{\bf 3}, \partial{\bf 3}, \pi_1[{\bf 3}\odot{\bf 3}]=\Xi \}$  \\
$\Gamma(4)$      & 16         & $\{1,\dots,\Xi,\partial^2{\bf 3}, \partial\Xi\}$   \\
$\vdots$         & $\vdots$   & \\
$\Gamma(m)$      & $m^2$      & $\{ \mbox{basis for } \Gamma(m-1)\}\cup\{\partial^{m+1}{\bf 3},\partial^{m-2}\Xi\}$ \\
\end{tabular}
\end{center}
\caption{Table of bases for Abelian functions associated with a
genus two curve.} \label{tab:g2EquiBases}
\end{table}

Now consider the algebraic Jacobian in genus two.
This is a locus of quadratic identities in ${\mathbb P}^{15}.$  The curve is, as before,
\[
y^2 = \lambda_6x^6 + 6\lambda_5x^5 + 15\lambda_4x^4 + 20\lambda_3x^3 + 15\lambda_2x^2 + 6\lambda_1x + \lambda_0
\]
Equivariant inhomogeneous coordinates on ${\mathbb P}^{15}$ are symmetric functions in a general divisor $((x_1,y_1),(x_2,y_2))$ of degree two on the curve having poles of order two or more on the divisor.  There are fifteen coordinates which fall into irreducible
$SL_2(\mathbb C)$ modules of dimensions one to five:
\[
{\bf 15}={\bf 5}_2\oplus{\bf 4}_3\oplus{\bf 3}_4\oplus{\bf 2}_5\oplus{\bf 1}_6
\]
It is a crucial observation that the divisor pole grading (denoted by the subscripts) is simply related to the module dimension,  i.e. that the dimension plus the degree at infinity is always equal to $7$.

The explicit coordinates on Jac($X$) are:
\begin{eqnarray}
{\bf 5}_2
&=&\left(\frac{1}{\delta^2},\frac{2(x_1+x_2)}{\delta^2},\frac{x_1^2+4x_1x_2+x_2^2}{\delta^2},\frac{2x_1x_2(x_1+x_2)}{\delta^2},\frac{x_1^2x_2^2}{\delta^2}\right),
\nonumber\\
{\bf 4}_3
&=&\left(\frac{y_1-y_2}{\delta^3},\frac{x_2y_1-x_1y_2}{\delta^3},\frac{x_2^2y_1-x_1^2y_2}{\delta^3},\frac{x_2^3y_1-x_1^3y_2}{\delta^3}\right),
\nonumber\\
{\bf 3}_4
&=&\left(\frac{I}{\delta},\frac{I(x_1+x_2)}{\delta},\frac{Ix_1x_2}{\delta}\right),
\nonumber\\
{\bf 2}_5
&=&\left(\frac{y_1I,_{x_1}+y_2I,_{x_2}}{\delta},\frac{x_2y_1I,_{x_1}+x_1y_2I_{x_2}}{\delta}\right),
\nonumber\\
{\bf 1}_6
&=&I^2,
\nonumber
\end{eqnarray}
where $\delta=x_1-x_2$ and $\delta^3 I$ is equal to the equivariant polar form of the curve, discussed earlier.
The vanishing of this polar form defines a cubic curve with order three contact at a point on the sextic curve above.

These coordinates are equivariant modifications of coordinates introduced in \cite{CasselsFlynn} and the Jacobian is the locus of a set of 72
identities quadratic therein. We seek such identities by tensoring
up the coordinate modules, decomposing into irreducibles,
identifying the divisor pole grading and then balancing both
singularities and dimension.
For example the projection
\[
\pi_5[ {\bf 3}_4 \odot {\bf 3}_4 - {\bf 5}_2 \odot {\bf 1}_6 ]
\]
onto a 5-dimensional component is, as is easily verified, vanishing.
This represents 5 quadratic identities. The other 67 can be equally
straightforwardly described and classified, \cite{CA_CF}.

The decompositions of tensor products of standard irreducible
modules ${\bf d}_i$ of dimension $d$ and pole order $i$ on the
divisor satisfy relations of the following type
\[
{\bf n}_s \odot {\bf m}_r \in \bigoplus_{p=|n-m|+1}^{|n-m|-1} {\bf p}_{\rho(n,s;m,r,p)}.
\]
The function $\rho(n,s;m,r,p)$ is the pole order of the irreducible component of dimension $p$ and it is an interesting problem to understand its structure.

These ideas also apply to tensoring of $\wp_{i_1i_2 \ldots i_n},$ which is why we could rewrite all of Table \ref{tab:g2Bases} in terms of modules in Table \ref{tab:g2EquiBases}.
By a simple counting argument one sees that
\begin{align*}
\dim\Gamma(m+1)-\dim\Gamma(m)&=\dim V(m+2)+\dim V(m-1), \\
\mbox{i.e.} \qquad (m+1)^2 - m^2 &= (m+2) + (m-1).
\end{align*}
The Hirota derivative, $\mathcal D,$ has been introduced earlier as
\[
{\mathcal D}(f,g)\rightarrow f'g-fg',
\]
and it can be characterised in terms of $SL_2$-modules  by the following equivariance property:
\[
\begin{array}{ccccc}
V\otimes V & \stackrel{\mathcal D}{\longrightarrow} & V\otimes V & \stackrel{Symm}{\longrightarrow} & V\nonumber\\
\downarrow & & \downarrow & & \downarrow\nonumber\\
V\otimes V & \longrightarrow & V\otimes V & \longrightarrow & V\nonumber
\end{array}
\]
In this (commuting) diagram the downward vertical arrow signifies
the $SL_2$-action. Extensions of this equivariant Hirota map to
higher order tensor products exist and our current work seeks to
understand basis generation (in ``appropriate'' gradings) by
studying equivariant resolutions of $SL_2$-modules of functions on
the Jacobian.

\section*{Acknowledgements}

The authors would like to thank the organizers of the conference, \emph{Finite Dimensional Integrable Systems in Geometry and Mathematical Physics 2011}, for the invitation to speak and financial assistance to attend this meeting.  The authors gave a joint talk at this meeting, from which this paper originated.

\section*{Contact Details}

\begin{tabular}{lcl}
Dr. Matthew England            & \hspace*{0.5in} & Dr. Chris Athorne           \\
Matthew.England@glasgow.ac.uk  &                 & Chris.Athorne@glasgow.ac.uk \\
\end{tabular}

\bigskip

\noindent \hspace*{0.05in} School of Mathematics \& Statistics, University of Glasgow, Glasgow, G12 8QW, UK.

\newpage

\begin{small}
\bibliography{GEF}{}
\bibliographystyle{plain}
\end{small}

\end{document}